\documentclass[conference, onecolumn]{IEEEtran}
\usepackage{setspace}
\usepackage{amsmath}
\usepackage{amssymb}
\usepackage{mathrsfs}
\usepackage{amsthm}
\usepackage{multirow}
\usepackage{color}
\usepackage{array}
\usepackage{url}
\usepackage{comment}
\usepackage{enumerate}
\usepackage{eucal}
\usepackage{bbm}
\usepackage{breqn}
\usepackage[top=0.73in, bottom=0.993in, left=0.673in, right=0.673in]{geometry}
\usepackage[ruled, lined, linesnumbered, commentsnumbered, longend]{algorithm2e}
\usepackage[normalem]{ulem}
\usepackage{cite}
\usepackage[table]{xcolor}
\usepackage{blkarray}

\usepackage{stmaryrd}

\usepackage{tikz}
\usetikzlibrary{arrows, patterns, shapes.arrows, decorations.pathmorphing}

\IEEEoverridecommandlockouts

\theoremstyle{plain}
\newtheorem{theorem}{Theorem}
\newtheorem{lemma}{Lemma}

\newtheorem{claim}[theorem]{Claim}
\theoremstyle{definition}

\newcommand{\NN}{\mathbb{N}}

\newcommand{\ZZ}{\mathbb{Z}}

\DeclareMathAlphabet{\mathbfsl}{OT1}{ppl}{b}{it} 

\newcommand{\vu}{\mathbfsl{u}}

\newcommand{\vw}{\mathbfsl{w}}
\newcommand{\vx}{\mathbfsl{x}}
\newcommand{\vy}{\mathbfsl{y}}

\newcommand{\cS}{\mathcal{S}}
\newcommand{\cA}{\mathcal{A}}
\newcommand{\cB}{\mathcal{B}}

\newcommand{\bfc}{{\boldsymbol c}}

\newcommand{\bfu}{{\boldsymbol u}}
\newcommand{\bfv}{{\boldsymbol v}}
\newcommand{\bfw}{{\boldsymbol w}}
\newcommand{\bfx}{{\boldsymbol x}}
\newcommand{\bfy}{{\boldsymbol y}}

\newcommand{\fail}{{\tt FAILURE}}

\title{\Huge{Sequence Reconstruction Problem for Sticky Insertion/Deletion Channels}}

\author{\textbf{Van Long Phuoc Pham}\IEEEauthorrefmark{1},
        \textbf{Yeow Meng Chee}\IEEEauthorrefmark{2},
        \textbf{Kui Cai}\IEEEauthorrefmark{2},
        and \textbf{Van Khu Vu}\IEEEauthorrefmark{3}\\[0.5mm]
\IEEEauthorblockA{\IEEEauthorrefmark{1} \small Brown University, USA \\[1mm]}
\IEEEauthorblockA{\IEEEauthorrefmark{2} \small Singapore University of Technology and Design, Singapore\\[1mm]}
\IEEEauthorblockA{\IEEEauthorrefmark{3} \small VinUniversity, Vietnam\\[0.5mm]}

{\footnotesize phuoc\_pham\_van\_long@brown.edu,~\{ymchee, cai\_kui\}@sutd.edu.sg, khu.vv@vinuni.edu.vn  }
}
\vspace{-6ex}

\begin{document}

\maketitle
\begin{abstract}
The sequence reconstruction problem for insertion/deletion channels has attracted significant attention owing to their applications recently in some emerging data storage systems, such as racetrack memories, DNA-based data storage. Our goal is to investigate the reconstruction problem for sticky-insdel channels where both sticky-insertions and sticky-deletions occur. 
If there are only sticky-insertion errors, the reconstruction problem for sticky-insertion channel is a special case of the reconstruction problem for tandem-duplication channel which has been well-studied. 
In this work, we consider the $(t,s)$-sticky-insdel channel where there are at most $t$ sticky-insertion errors and $s$ sticky-deletion errors when we transmit a message through the channel. For the reconstruction problem, we are interested in the minimum number of distinct outputs from these channels that are needed to uniquely recover the transmitted vector. We first provide an explicit formula for the minimum number of distinct outputs required.
Then, we discuss efficient algorithms to uniquely reconstruct the transmitted vector from erroneous sequences.
\end{abstract}

\section{Introduction}
Codes for the sticky insertion/deletion channel have been studied since the first paper by Levenshtein in 1965\cite{Lev65}. Several bounds on the channel capacity are known\cite{Mit08, CR19}. Many codes have been designed to correct sticky deletion/insertion errors with low redundancy\cite{DA10, TB08, TEB10, MV17}. Recently, codes correcting sticky insertion errors have attracted a significant attention owing to their new application in some data storage system, such as, racetrack memories \cite{CKVVY18, CKVVY18A} and DNA data storage systems \cite{MDSG17,JHSB17,KT18}. Especially, in \cite{MDSG17}, the authors modelled the nanopore sequencers as the sticky insertion/deletion channel and considered the sequence reconstruction problem for the channel. 

The sequence reconstruction problem, introduced by Levenshtein in 2001 in his seminal papers \cite{Lev01A,Lev01B}, considers a model where the sender transmits a vector $\bfx$ through a noisy channel a number of times (or through a number of noisy channels) and the receiver need to reconstruct the vector $\bfx$ from all received erroneous outputs. The transmitted vector $\bfx$ can be any vector or can be a codeword from a fixed codebook. In the reconstruction problem, we are interested in the minimum number of channels that are required for the receiver to reconstruct the transmitted vector uniquely and we are also interested in such an efficient algorithm to reconstruct the transmitted vector. To apply the reconstruction problem, we need to be able to send the message through the channel many times or we need to have many identical channels. Some emerging data storage systems, such as DNA-based data storage \cite{CGK12,YKGMZM15} 
and racetrack memories \cite{CKVVY18,CGVVY18,JB19}, are suitable for using the reconstruction problem. We note that the sequence reconstruction problem in the deletion channel has been well-studied\cite{GY18,SY21, CKY20}. Recently, in \cite{PK22}, the authors provided asymptotically exact estimate of the minimum number of channels required to reconstruct the transmitted vector in general case.
However, in the above data storage systems, both deletions and sticky insertions are common errors. We note that if only sticky-insertions occur in the channel, there are several related works on reconstruction problem for special versions of sticky-insertion channel, including block-version or tandem duplication\cite{MDSG17, YS19, WS22}. In this work, we focus on the reconstruction problem for the sticky-insdel channel, where both sticky-insertions and sticky-deletions occur.

We note that in the sticky-insdel channel, the number of runs of the transmitted vector is preserved, and the length of each run can be changed. Hence, we consider the transmitted vectors that have the same number of runs, say $r$. We assume that there are at most $t$ 
sticky-insertion errors and at most $s$ sticky-deletion errors when we transmit through a channel a vector $\bfx$ of $r$ runs. We aim to obtain exactly the minimum number of distinct outputs in these channels needed to reconstruct the transmitted vector. In the case that the transmitted vector $\bfx$ can be any binary vector with $r$ runs, we first provide a recursive formula to compute the exact value of the minimum number of distinct outputs needed. Then we use generating function and extract its coefficients to obtain an explicit formula for the exact number of required distinct erroneous sequences. 
Finally, we discuss efficient algorithms to recover the original vector from these erroneous sequences.

\section{Preliminaries} \label{sc:prelim}
Let $\mathbb{N}$ denote the set of natural numbers and $\mathbb{Z_+}$ denote the set of all positive integers. Let $\Sigma_q=\{0,1,\dotsc,q-1\}$ be the finite alphabet of size $q$, and the set of all strings of finite length over $\Sigma_q$ is denoted by $\Sigma_q^{*}$, while $\Sigma_q^n$ represents the set of all sequences of length $n$ over $\Sigma_q$. For two integers $i<j$, let $[i:j]$ denote the set $\{i,i+1,i+2, \ldots, j\}$. Let $[r]$ denote the set $\left\{ 1, \ldots, r \right\}$.
Let $\bfx = x_1x_2\ldots x_n \in \Sigma_q^n$ be a $q$-ary vector of length $n$. A \emph{run} is a maximal substring consisting of identical symbols and $\mathrm{r}(\bfx)$ denotes the number of runs of the sequence $\bfx$. We can  write $\bfx$ uniquely as $\bfx= c_1^{u_1}c_2^{u_2}c_3^{u_3}\dotsc c_r^{u_{r}}$, where $r=\mathrm{r}(\bfx)$ and $c_i \neq c_{i+1}$ for all $1 \leq i \leq r-1$. For example, $00311120=0^23^11^32^10^1$. Given a positive integer, $r$, let $P_q^r = \{c_1c_2\ldots c_r \in \Sigma_q^r : c_i \neq c_{i+1} \forall i \in [r-1]\}$ and $J_r^* =\{ \bfx \in \Sigma_q^* : \mathrm{r}(\bfx) = r\}$. 
We define the following map,
\begin{equation*}
    \Psi: J_r^* \to P_q^r \times \ZZ_+^r
\end{equation*}
such that if $\bfx= c_1^{u_1}c_2^{u_2}\ldots c_r^{u_{r}} \in J_r^* \subset \Sigma_q^*$ then $\Psi(\bfx)= (\bfc, \bfu)$ where $\bfc=c_1c_2\ldots c_r \in P_q^r \subset \Sigma_q^r$ and $\bfu= u_1u_2\ldots u_r \in \ZZ_+^r$.
For example, $\bfx=00311120 \in \Sigma_4^8$ and $\Psi(\bfx)=(\bfc,\bfu)$ where $\bfc=03120$ and $\bfu=21311.$ We note that $\Psi$ is a bijection and for each pair $(\bfc,\bfu) \in P_q^r \times \ZZ_+^r$, we can determine uniquely the sequence $\bfx \in J_r^*$ such that $\Psi(\bfx) = (\bfc,\bfu).$
In this paper, when describing the reconstruction algorithm, we also use a sub-map $\phi(\bfx) = \bfu$ specifically for the run-length.

Given the sequence $\bfx$, a sticky-insertion happens when a symbol in $\bfx$ is chosen, and another exact same symbol is inserted into $\bfx$, right next to this chosen token.
For example, given the sequence $\bfx=00311120$, if symbol 3 is inserted into $\bfx$ at the 3rd position to obtain $\bfy=003311120$ then we say a sticky-insertion occurs at the 3rd position. Let $\Psi(\bfx)=(\bfc,\bfu).$
If a sticky-insertion occurs in the sequence $\bfx$ to obtain $\bfy$ then $\Psi(\bfy) =(\bfc,\bfv)$ where there is only one index $i$ such that $v_i-u_i=1$. 
That is, there is no error in $\bfc$ and
the length of one run of $\bfx$ is increased by one. 
A sticky-deletion is an error that a symbol in a run of length at least two is deleted. We note that, once a sticky-deletion occurs in the sequence $\bfx$, the length of one run of $\bfx$ is decreased by one but it is not possible to delete the whole run. 

Given two positive integers $t,s,$ let the sticky-insdel\footnote{The word ``insdel'' is short for ``both insertion and deletion''.} ball $\cB_{t,s}(\bfx)$ be the set of all sequences obtained from $\bfx$ with at most $t$ sticky-insertions and $s$ sticky-deletions. 
More specifically, let $\cB_{t,s}(\bfx)$ denote the sticky-insdel error ball centered at $\bfx$ with radii $t$ and $s.$ A channel is called the $(t,s)$-sticky-insdel channel if a vector $\bfx$ is transmitted and a vector $\bfy \in \cB_{t,s}(\bfx)$ is received. Given two sequences, $\bfx_1$ and $\bfx_2$, we define the intersection of two sticky-insdel balls $\cB_{t,s}(\bfx_1,\bfx_2)= \cB_{t,s}(\bfx_1) \cap \cB_{t,s}(\bfx_2)$. Let $|\cB_{t,s}(\bfx_1,\bfx_2)|$ be the size of the intersection $\cB_{t,s}(\bfx_1,\bfx_2)$. In this work, we are interested in the problem of reconstructing a sequence $\bfx \in \Sigma_q^*$, knowing $M$ different sequences $\bfy_1,\ldots,\bfy_{M} \in \cB_{t,s}(\bfx)$. Since sticky-insertion and sticky-deletion errors do not change the number of runs of the sequence, we only consider $\bfx \in J_r^*$, given $r.$
Let
 $$N_{t,s}(r) = \max_{\bfx_1,\bfx_2 \in J_r^*} |\cB_{t,s}(\bfx_1,\bfx_2)|$$
be the maximal intersection of two sticky-insdel balls. If a sequence $\bfx$ is transmitted in the $(t,s)$-sticky-insdel channel, $N_{t,s}(r)+1$ distinct channel outputs are necessary and sufficient to guarantee a unique reconstruction of the transmitted sequence $\bfx.$
In the sequence reconstruction problem, we are interested in determining $N_{t,s}(r)$ and also in designing an algorithm to reconstruct $\bfx$ from $N_{t,s}(r)+1$ distinct sequences $\bfy_1,\ldots,\bfy_{N_{t,s}(r)+1} \in \cB_{t,s}(\bfx)$.

Given positive integers $r,t,s$ and $\bfu=u_1u_2\ldots u_r \in \ZZ_+^r,$ we define the following set 
 \begin{align*} 
    \cA_{t,s}(\bfu) = \left\{ (v_1, v_2 \ldots, v_r):  \vphantom{\sum_{i=1}^{r} s_i = s} \right.&\left. \sum_{i=1}^r \max \{ 0,u_i-v_i\} \le s, \vphantom{\sum_{i=1}^{r} s_i \leq s} \right. \nonumber \\
    & \left.    \sum_{i=1}^r \max \{ 0,v_i-u_i\} \le t \right\}.
\end{align*}
The set $\cA_{t,s}(\bfu)$ is called the asymmetric error ball centered at $\bfu$ with radii $t$ and $s.$ 

Given two sequences, $\bfx_1$ and $\bfx_2$, we have $\Psi(\bfx_1)=(\bfc_1,\bfu_1)$ and $\Psi(\bfx_2)=(\bfc_2,\bfu_2).$ We observe that $\bfy \in \cB_{t,s}(\bfx_1)$ if and only if $\Psi(\bfy)=(\bfc_1,\bfv)$ and $\bfv \in \cA_{t,s}(\bfu_1)$. 
Hence, $|\cB_{t,s}(\bfx_1)|=|\cA_{t,s}(\bfu_1)|$. Furthermore, $|\cB_{t,s}(\bfx_1,\bfx_2)|=|\cA_{t,s}(\bfu_1,\bfu_2)|$ where $\cA_{t,s}(\bfu_1,\bfu_2)=\cA_{t,s}(\bfu_1) \cap \cA_{t,s}(\bfu_2)$ is the intersection of two asymmetric error balls. So, 
$$N_{t,s}(r)=\max_{\bfu_1,\bfu_2 \in \ZZ_+^r} |\cA_{t,s}(\bfu_1,\bfu_2)|.$$
When only sticky-insertion occurs, that is $s=0$, an exact formula for  
$N_{t,0}(r)$ is known in \cite{YS19}. In particular, $N_{t,0}(r)=\binom{t+r-1}{r} +1$. Furthermore, it is straightforward to recover the original sequence from at least $N_{t,0}(r)$ distinct erroneous sequences. However, when both sticky-insertions and sticky-deletions occur in the channel, the value $N_{t,s}(r)$ is not known. In this paper, we provide an explicit formula of $N_{t,s}(r)$ in Section \ref{ssc:lower-bound-max-sticky-insdel-intersection}. Before that, we compute the maximum size of sticky-insdel balls in Section \ref{sc:max-sticky-insdel-size}. Finally, we
provide an algorithm to recover the original sequence in Section \ref{sc:reconstruction-uniqueness}.

\section{Maximum size of sticky insdel ball} \label{sc:max-sticky-insdel-size}
In this subsection, we will define a quantity $A_{t, s}(r)$ recursively, then claim that the maximal size of the asymmetric errors balls, $\max_{\bfu \in \ZZ_+^r} |\cA_{t, s}(\bfu)$, is exactly equal to $A_{t, s}(r)$. 
Then, equivalently, it means that the maximal size of the sticky insdel errors balls, $\max_{\bfx \in J_r^*} |\cB_{t, s}(\bfx)|$, is also equivalent to $A_{t, s}(r)$.
Later in Section~\ref{sc.required number}, this quantity $A_{t, s}(r)$ is further used as a building block for the formula to calculate the maximal intersection of two sticky-insdel balls $N_{t, s}(r)$.
For non-negative integers $t,s$, for all $r > 0,$ the recursive formula $A_{t,s}(r)$ is shown in Equation~\eqref{eq1-recursive-maxsize}, with the base cases being $A_{t, s}(1) = t + s + 1$ and $A_{t, s}(r) = 0$ for all $t, s < 0$.
\begin{align} 
    A_{t,s}(r) &= \sum_{i=1}^{t} A_{t-i,s}(r-1) + \sum_{i=1}^{s} A_{t,s-i}(r-1) \nonumber\\
    &+ A_{t,s}(r-1). \label{eq1-recursive-maxsize}
\end{align}
In this subsection, we will prove that $\max_{\bfu \in \ZZ_+^r} |\cA_{t,s}(\bfu)| = A_{t, s}(r)$ by first proving Lemma~\ref{lm1:max-ball-size-characterization}, which will imply that each $(t, s)$-sticky-insdel ball centered at $\bfu$, where all coordinates $u_i$ are at least $s+1$, has the maximum ball size.
Then, we show that these balls' size are exactly equal to $A_{t, s}(r)$ in Lemma~\ref{lm2:max-ball-size-large-tuple}, which would complete the section.
We formalize the main result of this section in Theorem~\ref{thm.insdel.ballsize}.
\begin{theorem} \label{thm.insdel.ballsize}
For all $t,s \in \NN$ and $r \in \ZZ_+$, we have:
\begin{equation}\label{eq2-max}
    \max_{\bfu \in \ZZ_+^r} |\cA_{t,s}(\bfu)| = A_{t,s}(r)
\end{equation}
\end{theorem}

We start the proof with Lemma~\ref{lm1:max-ball-size-characterization}, which implies full characterization for sticky insdel balls with maximum size.

\begin{lemma} \label{lm1:max-ball-size-characterization}
    For $\bfu =u_1u_2\ldots u_r \in \ZZ_+^r$, let $\bfu + s = (u_1+s, u_2+s, ..., u_r+s)$. The following inequality holds, $$\left| \cA_{t,s}(\bfu) \right| \le \left| \cA_{t,s}(\bfu+s) \right|.$$  
\end{lemma}
\begin{IEEEproof}
    Let $\bfw = (w_1, \ldots, w_r) \in \cA_{t,s}(\bfu)$, we will prove that $\bfw + s \in \cA_{t,s}(\bfu+s)$. 
    Indeed, by definition:
    \begin{equation*}
        \sum_{i=1}^r \max \{0, u_i - w_i \} \le s, \sum_{i=1}^r \max \{0, w_i - u_i \} \le t .
    \end{equation*}
    Hence, $\bfw +s$ also satisfies:
    \begin{align*}
        \sum_{i=1}^r \max \{0, (u_i + s) - (w_i + s) \} \le s, \\ \sum_{i=1}^r \max \{0, (w_i + s) - (u_i + s) \} \le t,
    \end{align*}

    which means that $\bfw+s \in \cA_{t,s}(\bfu+s)$. So we found an embedding $\bfw \in \cA_{t,s}(\bfu) \mapsto \bfw+s \in \cA_{t,s}(\bfu+s)$ from $\cA_{t,s}(\bfu)$ to $\cA_{t,s}(\bfu + s)$, and conclude that $\left| \cA_{t,s}(\bfu) \right| \le \left| \cA_{t,s}(\bfu+s) \right|$.
\end{IEEEproof}

Now, in Lemma~\ref{lm2:max-ball-size-large-tuple}, we calculate the exact size for $(t, s)$-sticky-insdel balls with centers being points with coordinates at least $s+1$, that are implied by Lemma~\ref{lm1:max-ball-size-characterization} to be maximum.

\begin{lemma} \label{lm2:max-ball-size-large-tuple}
	For $r \in \ZZ_+$ and $t, s \in \NN$, for each $\bfu=u_1u_2\ldots u_r \in \ZZ_+^{r}$ such that $u_i \ge s+1, \forall i \in [r]$, we have $\left|\cA_{t,s}(\bfu)\right| = A_{t,s}(r)$.
\end{lemma}

\begin{IEEEproof}
    We will prove using induction with the base case when $r = 1$. 
    
    \textit{Base case.} Let $\bfu \in \ZZ_+^{r}$ be arbitrary tuples of length $1$ and $\bfu = (a), a \ge s+1$. The $(t, s)$-sticky-insdel ball centered at $\bfu$ is then:
    \begin{equation}
        \cA_{t,s}(\bfu) = \left\{ \alpha: a-s \le \alpha \le a+t\right\}
    \end{equation}
    which has size $s+t+1 = A_{t,s}(1)$.

    \textit{Inductive step.} Assume that, for $r \ge 1$ and for all $t, s \ge 0$, given any $r$-tuple $\bfu = (u_1, \ldots, u_r) \in \ZZ^{r}$ such that $u_i \ge s+1, \forall i \in [r]$, we have:
	\begin{equation*}
		\left|\cA_{t,s}(\bfu)\right| = A_{t,s}(r).
	\end{equation*}
    Then given integers $t, s \ge 0$, for any tuple $\bfv = (v_1, \ldots, v_{r+1}) \in \ZZ^{r+1}$ such that $v_i \ge s+1, \forall i \in [r+1]$, we will prove that:
	\begin{equation*}
		\left|\cA_{t,s}(\bfv)\right| = A_{t,s}(r+1). 
	\end{equation*}
    Let $\bfv' = (v_2, \ldots, v_{r+1})$ be the sub tuple of $\bfv$ without its first entry. Since every entry of $\bfv'$ is at least $s+1$, by the induction hypothesis, we have $\left| \cA_{t,s}(\bfv') \right| = A_{t,s}(r)$.
    For each $1 \le i \le t$, let $\cS_{v_1+i} \subseteq \cA_{t,s}(\bfv)$ be the subset contains all tuples that starts with $v_1+i$. Then, since $\vw \in \cA_{t-i,s} (\bfv') \mapsto (u_1+i) \circ \vw \in \cS_{v_1+i}$ is an one-to-one bijection from $\cA_{t-i,s}(\bfv')$ to $\cS_{v_1+i}$, we have $\left| \cA_{t-i,s}(\bfv') \right| = \left| \cS_{v_1+i} \right|$. 
    Also, for any $1 \le i \le s$, the one-to-one bijection from $\cA_{t,s-i}(\bfv')$ to $\cS_{v_1-i}$ is $\vw \in \cA_{t,s-i}(\bfv') \mapsto (v_1-i) \circ \vw \in \cS_{v_1-i}$, which means that $\left| \cA_{t,s-i}(\bfv') \right| = \left| \cS_{v_1-i} \right|$. Similarly, we obtained $\left| \cA_{t,s}(\bfv') \right| = \left| \cS_{v_1} \right|$. 
    Since we can partition the $(t, s)$-sticky-insdel ball of $\bfv$ into disjoint subsets of tuples with the starting entries range from $v_1-s$ to $v_1+t$, we have:
    \begin{align*}
        \left|\cA_{t,s}(\bfv)\right| &= \sum_{i=1}^{t} \left| \cS_{v_1+i} \right| + \sum_{i=1}^{s} \left| \cS_{v_1-i} \right| + \left| \cS_{v_1} \right|\\
        &= \sum_{i=1}^{t} \left| \cA_{t-i,s}(\bfv') \right| + \sum_{i=1}^{s} \left| \cA_{t,s-i}(\bfv') \right| + \left| \cA_{t,s}(\bfv') \right| \\
        &= \sum_{i=1}^{t} A_{t-i,s}(r) + \sum_{i = 1}^{s} A_{t,s-i}(r) + A_{t,s}(r) \\
        &= A_{t,s}(r+1).
    \end{align*}
\end{IEEEproof}

From Lemma \ref{lm1:max-ball-size-characterization} and Lemma \ref{lm2:max-ball-size-large-tuple}, given $r \in \NN$, $t,s \in \ZZ_+$, for all $\bfu \in \ZZ_+^r$, we see that $|\cA_{t,s}(\bfu)| \leq A_{t,s}(r)$ and thus $\max_{\bfu \in \ZZ_+^r} |\cA_{t,s}(\bfu)| \leq A_{t,s}(r)$.
From Lemma \ref{lm2:max-ball-size-large-tuple}, we can obtain the lower bound $\max_{\bfu \in \ZZ_+^r} |\cA_{t,s}(\bfu)| \geq A_{t,s}(r)$, which would be combined with the previous inequality to obtain Theorem~\ref{thm.insdel.ballsize}.
To compute an explicit formula of $A_{t,s}(r)$, we first provide a generating function of $A_{t,s}(r)$ and then extract its coefficients. 
Given $r \in \ZZ_+,$ let  $f_r(x, y) = \sum_{t, s \ge 0} A_{t,s}(r) x^t y^s$ be a generating function of $A_{t,s}(r)$. 
Using Equation \ref{eq1-recursive-maxsize}, we obtain
\begin{align*}
    &f_{r+1}(x, y) \\
    &= \sum_{t, s \ge 0} \left( \sum_{i=1}^t A_{t-i,s} (r) + \sum_{i=1}^s A_{t,s-i} (r) + A_{t,s} (r) \right) x^t y^s
\end{align*}
which can be rewritten as
\begin{equation}\label{eq-generatingfunction-r=1}
  f_{r+1}(x, y) = \frac{1-xy}{(1-x)(1-y)} f_r(x, y).
\end{equation}
Furthermore, since $A_{t,s}(1) =t+s+1$, we obtain
\begin{equation}
    f_1(x,y)= \sum_{t, s \ge 0} (1 + t + s) x^t y^s = \frac{1-xy}{(1-x)^2 (1-y)^2}.
\end{equation}
Therefore, we obtain a formula of the generating function of $A_{t,s}(r)$ as follows.
\begin{equation}\label{eq-generating-final-r}
    f_r(x, y) = \frac{(1-xy)^r}{(1-x)^{r+1} (1-y)^{r+1}}.
\end{equation}
Extracting coefficient of the generating function (\ref{eq-generating-final-r}), we obtain an explicit formula of $A_{t,s}(r)$ as follows.
\begin{equation}\label{eq.Atsr}
    A_{t,s}(r) = \sum_{i=0}^{\min (t, s, r)} \binom{r+t-i}{r} \binom{r+s-i}{r} \binom{r}{i} (-1)^{i}.
\end{equation}

\section{Minimum number of distinct outputs} \label{sc.required number}
In this section, our goal is to provide an exact formula of the value of $N_{t,s}(r) =\max_{\bfu,\bfv \in \ZZ_+^r} |\cA_{t,s}(\bfu,\bfv)| +1$.
The main result in this subsection is stated formally as follows.

\begin{theorem}\label{thm.insdel.exact}
Given $t,s,r$, then $N_{t, s}(r) = M$, with $M$ defined as:
    \begin{equation}
        M= \sum_{i=1}^{t} A_{t-i,s}(r-1) + \sum_{i=1}^{s} A_{t,s-i}(r-1)+1.
    \end{equation}
    where $A_{t,s}(r)$ is defined in Equation (\ref{eq.Atsr}).
\end{theorem}
To prove Theorem \ref{thm.insdel.exact}, we prove the lower bound of $N_{t,s}(r)$ in Subsection \ref{ssc:lower-bound-max-sticky-insdel-intersection} and show the upper bound of $N_{t,s}(r)$ in Subsection \ref{sc:reconstruction-uniqueness}. We note that from Theorem \ref{thm.insdel.exact}, Equation~\eqref{eq1-recursive-maxsize}, and Equation~\eqref{eq.Atsr} we can obtain an exact formula of the minimum number of distinct outputs required for uniquely reconstruction as follows.
{\small{
\begin{align*}
        &N_{t,s}(r)\\ &=1+ \sum_{i=0}^{\min (t, s, r)} \binom{r+t-i}{r} \binom{r+s-i}{r} \binom{r}{i} (-1)^{i}\\
        &- \sum_{i=0}^{\min (t, s, r-1)} \binom{r-1+t-i}{r} \binom{r-1+s-i}{r} \binom{r-1}{i} (-1)^{i}.
\end{align*}
}
}
\subsection{Lower bound}\label{ssc:lower-bound-max-sticky-insdel-intersection}
In this section, our goal is to prove the lower bound of the value of  $N_{t,s}(r)$ as in Theorem \ref{thm.insdel.exact}. In particular, we will prove the following result,
    \begin{equation}\label{thm.insdel.min}
        N_{t,s}(r) \geq \sum_{i=1}^{t} A_{t-i,s}(r-1) + \sum_{i=1}^{s} A_{t,s-i}(r-1)+1.
    \end{equation}
To prove Equation~\eqref{thm.insdel.min}, we show that there are two vectors $\bfu_0,\bfv_0 \in \ZZ_+^r$ such that
    \begin{equation}\label{eq.insdel.min}
        |\cA_{t,s}(\bfu_0, \bfv_0)| \ge  \sum_{i=1}^{t} A_{t-i,s}(r-1) + \sum_{i=1}^{s} A_{t,s-i}(r-1).
    \end{equation}
Indeed,  for $r \ge 3, n \ge (s+1) * r+1$, let $n = (s+1) * r + k$. Let $\bfu_0 = (u_1, \ldots, u_r)$ such that $u_1 = s+2, u_2 = s+1, u_3 = s+k$ and
       $ u_i = s+1, \forall i > 3$. 
Let $\bfv_0 = (v_1, \ldots, v_r)$ such that
   $ v_1 = s+1, v_2 = s+2, v_3 = s+k,$ and
       $v_i = s+1, \forall i > 3.$ 
We now compute $|\cA_{t,s}(\bfu_0, \bfv_0) |$ to obtain Inequality (\ref{eq.insdel.min}).

      Let $\cS = \cA_{t,s}(\bfu_0, \bfv_0)$. We can decompose $\cS$ into disjoint subsets $\cS_{\alpha} \subseteq \cS$ which contains all the tuples in $\cS$ that begin with $\alpha$. Note that since there are $t$ insertions and $s$ deletions, $\alpha$ can only range from $u_1-s = 2$ to $v_1+t = t+s+1$.
    \begin{align*}
        \left| \cS \right| &= \sum_{\alpha=2}^{t+s+1} \left| \cS_{\alpha} \right| = \sum_{\alpha=2}^{s+1} \left| \cS_{\alpha} \right| + \sum_{\alpha=s+2}^{t+s+1} \left| \cS_{\alpha} \right| \\ &= \sum_{i=1}^{s} \left| \cS_{u_1-i} \right| + \sum_{i=1}^{t} \left| \cS_{v_1+i} \right|.
    \end{align*}

    We now show that, for every $1 \le i \le t$, the size $|\cS_{v_1+i}| \ge  A_{t-i,s}(r-1)$. Let $\bfw_0 = (w_1, \ldots, w_r) \in \ZZ_+^r$ such that $w_{1} = s+i+1$, $w_{2} = s+2, w_{3} = s+k$, and $w_{j} = s+2$ for all $j>1$.
   From $\bfu_0$ we can create $\bfw_0$ by adding $i-1$ into $u_1$ and adding $1$ into $u_2$, whereas from $\bfv_0$ we can create $\bfw_0$ by adding $i$ into $v_1$. Let $\bfw_0' = (w_2, \ldots, w_r)$ be the sub-vector of $\bfw_0$ beginning from the second element of $\bfw_0$, then for each vector in $\cA_{t-i,s}(\vw_0')$, prepending $(s+1+i)$ into it gives us a unique tuple in $\cS_{s+1+i}$. This shows 
   that $$ |\cS_{v+i}| = \left| \cS_{s+i+1} \right| \ge \left| \cA_{t-i,s}(\bfw_0') \right| = A_{t-i,s}(r-1).$$
    
    For each $1 \le i \le s$, in a similar manner, we can show that $\left| \cS_{u_1-i} \right| = A_{t,s-i}(r-1)$. Then after summing up, we get:
    \begin{align*}
        \left| \cS \right|  &= \sum_{i=1}^{t} \left| \cS_{v_1+i} \right| + \sum_{i=1}^{s} \left| \cS_{u_1-i} \right| \\
        &\ge \sum_{i=1}^{t} A_{t-i,s}(r-1) + \sum_{i=1}^{s} A_{t,s-i}(r-1)
    \end{align*}
    Hence, Inequality~\eqref{eq.insdel.min} holds and thus, the lower bound of the minimum number of distinct outputs is given by Equation~\eqref{thm.insdel.min}. 
In the next section, we show that the above lower bound is tight by providing an efficient algorithm to reconstruct the transmitted vector from $M$ distinct outputs, where $M=N_{t,s}(r)$ as in Theorem \ref{thm.insdel.exact}.

\subsection{Upper Bound} \label{sc:reconstruction-uniqueness}
In this subsection, our goal is to prove the upper bound of $N_{t,s}(r)$, that is $N_{t,s}(r) \leq M.$ 
To obtain the upper bound, we show that it is possible to reconstruct $\bfx$ uniquely from $M$ distinct tuples $\bfy_1, \ldots, \bfy_M$, where $\bfy_i \in \cB_{t,s}(\bfx)$ for all $1 \leq i \leq M$.  
We first present the following properties of the vector $\bfu=\phi(\bfx)$.
\begin{lemma} \label{lm:reconstruction-conditions}
    For each $i$, let $\bfv_i = (v_{i,1}, \ldots, v_{i, r}) = \phi(\bfy_i)$. Let $\bfx$ be a possible transmitted vector and $\bfu = (u_1, \ldots, u_r) = \phi(\bfx)$. For each $j \in [r]$, let $a_j = \min_{i=1}^M v_{i, j}$, $b_j = \max_{i=1}^M v_{i, j}$, and $c_{j, k}$ be the number of indices $i \in [M]$ such that $v_{i, j} = k$. We note that $\sum_{k = a_j}^{b_j} c_{j, k} = M$.
    Then, $u_j$ satisfies the following properties:
    \begin{enumerate}
        \item \label{it:condition-0} $ b_j-t \leq u_j \leq a_j+s$.
        \item \label{it:condition-1} For any $a_j \le k < u_j$, we have:
        \begin{equation}
            c_{j, k} \le A_{t,s-u_j+k}(r-1). \label{eq:condition-1}\\
        \end{equation}
        \item \label{it:condition-2} For any $u_j < k \le b_j$, we have:
        \begin{equation}
            c_{j, k} \le A_{t-k+u_j,s}(r-1). \label{eq:condition-2}
        \end{equation}
    \end{enumerate}
\end{lemma}
\begin{IEEEproof}
We now prove the first claim, that is $ b_j-t \leq u_j \leq a_j+s$.
   For each $i \in [M]$, since $\bfv_i \in \cA_{t,s}(\bfu)$, by definition we have:
        \begin{equation*}
            \sum_{j=1}^r \max \{ 0,u_j-v_{i, j}\} \le s.
        \end{equation*}
        Hence $u_j - v_{i, j} \le s$, that is $u_j \le v_{i, j} + s$. Since this is true for all $i$ from $1$ to $M$, 
        we have $u_j\le a_j + s$. Similarly, we also obtain $u_j \ge b_j - t$. Hence, the first claim is proven.

        Before we prove Claim 2 and Claim 3, we define the vector $\bfu'$ be the vector obtained from $\bfu$ after removing the $j$-the entry, that is $\bfu'=(u_1,\ldots,u_{j-1},u_{j+1},\ldots,u_r)$.
        Let $\bfv_{m_1}, \ldots, \bfv_{m_{c_{j, k}}}$ be all the vectors in $\left\{ \bfv_1, \ldots, \bfv_M \right\}$ such that:
        \begin{equation*}
            v_{m_i, j} = k, \forall i \in [c_{j, k}].
        \end{equation*}
        
        For each $i \in [c_{j,k}]$,
        \begin{equation*}
            \bfv_{m_{i}} = (v_{m_i, 1}, \ldots, v_{m_i, j-1}, k, v_{m_i, j+1}, \ldots, v_{m_i, r}).
        \end{equation*}
        Let $\bfv'_{m_{i}}$ be a subsequence obtained from $\bfv_{m_i}$ after removing the $j$-th entry, that is
        \begin{equation*}
            \bfv'_{m_{i}} = (v_{m_i, 1}, \ldots, v_{m_i, j-1}, v_{m_i, j+1}, \ldots, v_{m_i, r}).
        \end{equation*}

        By definition of the asymmetric error ball, we have
        \begin{align}
            &\sum_{l=1}^{r} \max \{ 0, u_l - v_{m_i, l} \} \le s, \forall i \in [c_{j, k}], \text{ and} \label{eq1.s} \\
            &\sum_{l = 1}^j \max \{ 0, v_{m_i, l} - u_l \} \le t, \forall i \in [c_{j, k}]. \label{eq1.t}
        \end{align}

   Now, to prove Claim 2, we consider the case $a_j \le k < u_j$, that is $u_j-k>0$. From Equation~\eqref{eq1.s}, we obtain        \begin{equation*}
            \sum_{l \neq j} \max \{ 0, u_l - v_{m_i, l} \} \le s - u_j + k, \forall i \in [c_{j, k}].
        \end{equation*}
   From Equation~\eqref{eq1.t}, we obtain
           \begin{equation*}
            \sum_{l \neq j} \max \{ 0, v_{m_i, l} - u_l \} \le t, \forall i \in [c_{j, k}].
        \end{equation*}
        
        From the above two equations, we obtain $\bfv'_{m_i} \in \cA_{t,s-u_j+k} (\bfu'), \forall i \in [c_{j, k}]$. We note that all $c_{j,k}$ vectors $\bfv'_{m_i}$ are distinct. 
        Therefore $ A_{t,s-u_j+k} \geq |\cA_{t,s-u_j+k} (\bfu')| \geq c_{j, k}$. So, Claim 2 is proven.

        To prove Claim 3, we consider the case $u_j < k \le b_j$, that is $u_j-k <0$. From Equation~\eqref{eq1.t}, we obtain
                \begin{equation*}
            \sum_{l \neq j} \max \{ 0, v_{m_i, l} - u_l \} \le t - k + u_j, \forall i \in [c_{j, k}].
        \end{equation*}
        From Equation~\eqref{eq1.s}, we obtain
                \begin{equation*}
            \sum_{l \neq j} \max \{ 0, u_l - v_{m_i, l} \} \le s, \forall i \in [c_{j, k}].
        \end{equation*}
        From the above two equations, we obtain $\bfv'_{m_i} \in \cA_{t-k+u_j,s} (\bfu'), \forall i \in [c_{j, k}]$. Therefore 
        $c_{j, k} \le |\cA_{t-k+u_j,s}(\bfu')|  \le A_{t-k+u_j,s}(r-1)$.
So, Claim 3 is proven.
\end{IEEEproof}
From Lemma \ref{lm:reconstruction-conditions}, we see that if $\bfx$ is the original vector then vector $\bfu =\phi(\bfx)$ must satisfy three properties. We now show that there is only one vector that satisfies all these properties and thus there is only one vector $\bfx$ such that $\bfy_i \in \cB_{t,s}(\bfx)$ for all $1\leq i \leq M.$
 We state the result formally as follows.

\begin{lemma}\label{lm:reconstruction-uniqueness}
Given all parameters as in Lemma \ref{lm:reconstruction-conditions}, if there is $u_j'$ that satisfies all three properties in Lemma 4, then $u_j'=u_j$.
\end{lemma}

\begin{IEEEproof}
We see that $b_j - a_j \le s + t$. If $b_j - a_j = t + s$ and both $u_j, u_j'$ satisfy Property 1, then $u_j'=u_j=b_j-t=a_j+s$. Thus the lemma is proven in this case. 

 We now consider the case $b_j - t < a_j + s$ and prove $u_j=u_j'$ by a contradiction. Without loss of generality, we assume that $u_j<u_j'$ and both satisfy all three properties in Lemma \ref{lm:reconstruction-conditions}.
    \begin{itemize}
        \item For every $a_j \le k \le u_j$, since $k \le u_j < u_j'$, from \eqref{eq:condition-1} we have:
        \begin{equation*}
            c_{j, k} \le A_{t,s-u_j'+k} (r-1).
        \end{equation*}
        Hence:
        \begin{equation*}
            \sum_{k=a_j}^{u_j} c_{j, k} \le \sum_{k=a_j}^{u_j} A_{t,s-u_j'+k} (r-1) 
            \le \sum_{i=1}^{s} A_{t,s-i} (r-1).
        \end{equation*}
        Here the last inequality follows the fact that $u_j'-u_j \ge 1$ and $u_j' - a_j \le s$.
        \item For every $u_j < k \le b_j$, from \eqref{eq:condition-2} we have:
        \begin{equation*}
            c_{j, k} \le A_{t-k+u_j,s} (r-1).
        \end{equation*}
        Hence:
        \begin{align*}
            \sum_{k=u_j+1}^{b_j} c_{j, k} &\le \sum_{k=u_j+1}^{b_j} A_{t-k+u_j,s} (r-1) \\
            &\le \sum_{i=1}^{t} A_{t-i,s} (r-1).
        \end{align*}
        Here the last inequality follows the fact that $b_j-u_j \le t$.
    \end{itemize}

    Summing up, we obtain $$\sum_{k=a_j}^{b_j} c_{j, k} \le \sum_{i=1}^{s} A_{t,s-i}(r-1) + \sum_{i=1}^{t} A_{t-i,s}(r-1) < M.$$
This is a contradiction as we know that $\sum_{k=a_j}^{b_j} =M$ in Lemma \ref{lm:reconstruction-conditions}. We conclude the lemma since $u_j=u_j'$.
\end{IEEEproof}


\section{Reconstruction Algorithm}
From previous section, we showed that it is possible to uniquely recover the original sequence $\bfx$ from $M$ distinct erroneous sequences $\bfy_1,\bfy_2,\ldots,\bfy_M$ where $M$ is $N_{t,s}(r)$ as in Theorem \ref{thm.insdel.exact} and each $\bfy_i$ is obtained from $\bfx$ after at most $t$ sticky-insertions and $s$ sticky-deletions. 
In this section, we briefly discuss the reconstruction algorithms and analyze the complexity of these algorithms. 
Section~\ref{sc:reconstruction-uniqueness} has already established the theoretical basis for there construction algorithm.
 For each index $j$ from $1$ to $r$, we define $a_j, b_j$ as in Lemma \ref{lm:reconstruction-conditions}, then obtain the only value of $u_j \in [b_j-t, a_j+s]$ that satisfies the conditions as shown in Lemma \ref{lm:reconstruction-conditions}. We can reconstruct $\bfx$ from $\bfy_1, \ldots, \bfy_M$ following Algorithm \ref{algo:reconstruction}. The algorithm requires some subroutines, including $\mathrm{CheckValid}$ Algorithm.

\begin{algorithm}[h!] \label{algo:check-valid}
    \SetKwInOut{KwIn}{Input}
    \SetKwInOut{KwOut}{Output}
    \SetKw{KwExists}{exists}
    \SetKw{KwAnd}{and}
    \SetKw{KwBreak}{break}

    \KwIn{$\bfy_1, \ldots, \bfy_M$}
    \KwOut{$\vx$ or $\fail$}

    $\mathrm{formation}_1 \gets [y_{1, 1}]$ \\

    $\mathrm{pos} \gets 1$
    
    \For{$j = 2$ \KwTo $n$}{
        \If{$y_{1, j} \neq y_{1, j-1}$}{
        {
        $\mathrm{pos}++$
        
        $\mathrm{formation}_{\mathrm{pos}} \gets y_{1, j}$}
        }
    }

    \For{$i = 2$ \KwTo $m$}{
        $\mathrm{check\_pos} \gets 1$
        
        \lIf{$y_{i, 1} \neq \mathrm{formation}_1$}{\KwRet{$\fail$}}

        \For{$j = 2$ \KwTo $n$}{
            \If{$y_{i, j} \neq y_{i, j-1}$}{
                $\mathrm{check\_pos}++$

                \lIf{$\mathrm{check\_pos} > \mathrm{pos}$}{\KwRet{$\fail$}}

                \lIf{$y_{i, j} \neq \mathrm{formation}_{\mathrm{check\_pos}}$}{\KwRet{$\fail$}}
            }
        }
    }

    \KwRet{$\mathrm{SUCCESS}$}
    \caption{Check validity of inputs $\bfy_1, \ldots, \bfy_M$.}
\end{algorithm}

In Algorithm \ref{algo:reconstruction}, we first check the validity of $\bfy_1, \ldots, \bfy_M$ by seeing if they follow the same run formation. After passing the check, we iterate through each run index from $1$ to $r$, do processing (obtaining $a_i, b_i$, as well as all the counts), then find the satisfied $u_j$.

\begin{algorithm}[h!] \label{algo:reconstruction}
    \SetKwInOut{KwIn}{Input}
    \SetKwInOut{KwOut}{Output}
    \SetKw{KwExists}{exists}
    \SetKw{KwAnd}{and}
    \SetKw{KwBreak}{break}

    \KwIn{$\bfy_1, \ldots, \bfy_M$}
    \KwOut{$\vx$ or $\fail$}

    $\mathrm{CheckValid}(\bfy_1, \ldots, \bfy_M)$
    
    \lFor{$i = 1$ \KwTo $M$}{
        $\bfv_i = (v_{i1}, \ldots, v_{ir}) \gets \phi(\vy_i)$
    }

    \For{$j = 1$ \KwTo $r$}{
        $a_j \gets \min_i(v_{i, j})$ \\
        $b_j \gets \max_i(v_{i, j})$ \\
        \lIf{$b_j - a_j > s + t$}{\KwRet{$\fail$}}
        \lElseIf{$b_j - a_j = s + t$}{
            $u_j \gets a_j + s$ \\
        }
        \Else{
            $(c_{j, a_j}, c_{j, a_j+1}, \ldots, c_{j, b_j}) \gets \mathrm{FindFrequency}(a_j, b_j, (v_{1, j}, \ldots, v_{M, j})$ \\
            $u_j \gets \mathrm{FindSatisfiedValue}(a_j, b_j, (c_{j, a_j}, c_{j, a_j+1}, \ldots, c_{j, b_j}))$ \\
            \lIf{$u_j = \fail$}{return $\fail$}
        }
    }
    \KwRet{$\phi^{-1}(\vu)$}
    \caption{Reconstruction algorithm from $M = N_{t,s}(r)$ sticky insdel channel outputs to a tuple in $\NN^{r}$}
\end{algorithm}

\textit{Time Complexity.} Checking validity of $\bfy_1, \ldots, \bfy_M$ and obtaining run lengths require linear scans, which costs $O(Mn)$ time. After the preprocessing, we iterate through each run index from $1$ to $r$, where obtaining $a_j$ and $b_j$ costs $O(M)$, obtaining frequencies also costs $O(M)$, and the complexity for $\mathrm{FindSatisfiedValue}$ would be discussed in Section \ref{ssc:naive-approach} and Section \ref{ssc:optimized-approach}.

\subsection{Naive Approach} 
\label{ssc:naive-approach}
There is a naive approach to find $u_j$ that satisfies the conditions in Lemma \ref{lm:reconstruction-conditions} by simply looping through all $u_j \in [a_j, b_j]$, and use another loop to check for the correctness of $u_j$. The naive subroutine is shown in Algorithm \ref{algo:naive-find-original}.

\begin{algorithm}[h!] \label{algo:naive-find-original}
    \SetKwInOut{KwIn}{Input}
    \SetKwInOut{KwOut}{Output}
    \SetKw{KwExists}{exists}
    \SetKw{KwAnd}{and}
    \SetKw{KwBreak}{break}

    \KwIn{$a_j, b_j, (c_{j, a_j}, c_{j, a_j+1}, \ldots, c_{j, b_j})$}
    \KwOut{$u_j$ or $\fail$}

    \For{$u_j = b_j - t$ \KwTo $a_j + s$}{
        $\mathrm{flag} \gets 1$
        
        \For{$k = a_j$ \KwTo $u_j-1$}{
            \lIf{$c_{j, k} > A_{t,s - u_j + k}(r-1)$}{$\mathrm{flag} \gets 0$}
        }
        
        \For{$k = u_j+1$ \KwTo $b_j$}{
            \lIf{$c_{j, k} > A_{t - k + u_j,s}(r-1)$}{$\mathrm{flag} \gets 0$}
        }

        \lIf{$\mathrm{flag}$}{\KwRet{$u_j$}}
        \lElse{continue}
    }
    
    \KwRet{$\fail$}
    \caption{Naive algorithm $\mathrm{FindSatisfiedValue}$}
\end{algorithm}

\textit{Time Complexity.} Looping through all $u_j$ from $b_j-t$ to $a_j+s$ and then $k$ from $a_j$ to $b_j$ takes $(a_j - b_j + s + t)(b_j - a_j) \le O((s+t)^2)$ iterations.

\subsection{Optimized Approach}
\label{ssc:optimized-approach}
In this subsection, we provide an efficient two-pointer algorithm to find the unique value $u_j$ satisfying Lemma \ref{lm:reconstruction-conditions}. 

Claim \ref{lm:slice-1} further suggests that there exists $\alpha_j \in [b_j-t, a_j+s]$ such that $x$ satisfies \eqref{eq:condition-1} for all $x \in [b_j-t, \alpha_j]$, and does not satisfy \eqref{eq:condition-1} otherwise.
Similar to Claim \ref{lm:slice-1}, Claim \ref{lm:slice-2} "slice" the interval $[b_j-t, a_j+s]$ on the other direction, by showing that there exists $\beta_j \in [b_j-t, a_j+s]$ such that $x$ satisfies \eqref{eq:condition-2} for all $x \in [\beta_j, a_j+s]$, and does not satisfy \eqref{eq:condition-2} otherwise. 

\begin{claim}\label{lm:slice-1}
    If $\alpha \in [b_j - t, a_j + s]$ does not satisfy \eqref{eq:condition-1}, then $u$ does not satisfy \eqref{eq:condition-1} for all $u > \alpha$.
\end{claim}

\begin{proof}
    Since $\alpha$ does not satisfy \eqref{eq:condition-1}, there exists $k$ such that $a_j \le k < \alpha_j$, and:
    \begin{equation*}
        c_{j, k} > A_{t,s-\alpha_j+k} (r-1).
    \end{equation*}
    For $u > \alpha$, we also see that $k$ and $u$ does not satisfy \eqref{eq:condition-1}. In other word:
    \begin{equation*}
        c_{j, k} > A_{t,s-u+k} (r-1),
    \end{equation*}
    since $A_{t,s-\alpha+k} (r-1) > A_{t,s-u+k} (r-1)$ due to the fact that $s-\alpha+k > s-u+k$. 
    We conclude that for all $u > \alpha$, $u$ does not satisfy \eqref{eq:condition-1}.
\end{proof}

From the proof of claim \ref{lm:slice-2}, we obtained these observations:
\begin{itemize}
    \item If the check for value $u$ and "check position" $k$ with $k < u$ satisfies condition \eqref{eq:condition-1}, we can skip all the checks for pairs $(k, u')$ with $k < u' < u$, since they all satisfy \eqref{eq:condition-1}.
    \item If the check for value $u$ and "check position" $k$ with $k < u$ fails \eqref{eq:condition-1}, we do not need to check back if $(k, u')$ satisfies \eqref{eq:condition-1} for any $u' > u$, since they would all fail.
\end{itemize}

The observations suggested the two-pointer Algorithm \ref{algo:optimized-find-original}. We also stated the proof for Claim \ref{lm:slice-2} here for completeness.

\begin{algorithm}[h!] \label{algo:optimized-find-original}
    \SetKwInOut{KwIn}{Input}
    \SetKwInOut{KwOut}{Output}
    \SetKw{KwExists}{exists}
    \SetKw{KwAnd}{and}
    \SetKw{KwBreak}{break}

    \KwIn{$a_j, b_j, (c_{j, a_j}, c_{j, a_j+1}, \ldots, c_{j, b_j})$}
    \KwOut{$u_j$ or $\fail$}

    $u_j \gets a_j + s, k \gets a_j + s - 1$ \\
    \While{$\mathrm{true}$}{
        \lIf{$k = a_j$}{
            \KwRet{$u_j$}
        }
        \lElseIf{$k = u_j$}{
            $k--$;
        }
        \uElse{
            \lIf{$c_{j, k} > A_{t,s-u_j+k} (r-1)$}{$u_j--$}
            \lElse{ $k--$}
        }
    }
    \KwRet{$\fail$}

    \caption{Optimized algorithm $\mathrm{FindSatisfiedValue}$}
\end{algorithm}

\begin{claim} \label{lm:slice-2}
    If $\beta \in [b_j-t, a_j+s]$ does not satisfy \eqref{eq:condition-2}, then $u$ does not satisfy \eqref{eq:condition-2} for all $u < \beta$.
\end{claim}

\begin{proof}
    Since $\beta$ does not satisfy \eqref{eq:condition-2}, there exists $i$ such that $\beta < k \le b_j$ and:
    \begin{equation*}
        c_{j, k} > A_{t-k+\beta,s} (r-1).
    \end{equation*}
    For $u < \beta$, we also see that $k$ and $u$ does not satisfy \eqref{eq:condition-2}. In other word:
    \begin{equation*}
        c_{j, k} > A_{t-k+u,s} (r-1).
    \end{equation*}
    We note that $A_{t-k+\beta,s} (r-1) > A_{t-k+u,s} (r-1)$ since $t-k+\beta > t-k+u$.
    We conclude that for all $u < \beta$, $u$ does not satisfy \eqref{eq:condition-2}.
\end{proof}

\section{Conclusion}
In this paper, we have studied the sequence reconstruction problem for sticky-insdel channel. The first question of interest is to find the minimum number of distinct erroneous outputs needed for the receiver to uniquely reconstruct the transmitted vector. In this work, we consider the case that there are at most $t$ sticky-insertion errors and $s$ sticky-deletion errors.
We used some enumerative combinatorics techniques to provide an explicit formula of required number. Then, we also provided efficient algorithms to uniquely recover the original vector from enough erroneous sequences. 

\section{Acknowledgement}
We thank Manabu Hagiwara for the insightful discussion about the explicit formula of $A_{t, s}(r)$ in Equation~\eqref{eq.Atsr}.

\newpage





\end{document}